\documentclass{llncs}

\usepackage{geometry}

\usepackage{tikz}

\usepackage{latexsym}
\usepackage[intlimits]{amsmath}
\usepackage{amsfonts}
\usepackage{txfonts}
\usepackage{dsfont}

\usepackage[ruled,vlined,linesnumbered,lined]{algorithm2e}

\newtheorem{fact}[theorem]{Fact}  %
\newenvironment{reftheorem}[2]{\begin{trivlist}
\item[\hskip \labelsep {\bfseries\scshape #1}\hskip \labelsep {\bfseries #2.}]\itshape}{\end{trivlist}\slshape}

\newcommand{\A}{\mathcal{A}}
\newcommand{\G}{\mathcal{G}}
\renewcommand{\H}{\mathcal{H}}
\newcommand{\M}{\mathcal{M}}

\newcommand{\eps}{\varepsilon}

\newcommand{\genTran}[2]{%
    {}\mathchoice%
    {\stackrel{#1}{#2}}
    {\mathop {\smash{#2}}\limits^{\vrule width 0pt height 0pt depth 4pt\smash{#1}}}
    {\stackrel{#1}{#2}}
    {\stackrel{#1}{#2}}
{}}
\newcommand{\tran}[1]{\genTran{#1}{\rightarrow}}
\newcommand{\btran}[1]{\genTran{#1}{\hookrightarrow}}
\newcommand{\ctran}[1]{\genTran{#1}{\mapsto}}

\newcommand{\run}{\mathit{Run}}

\newcommand{\Val}{\mathit{Val}}
\newcommand{\Prob}{\mathit{Prob}}

\renewcommand{\Pr}{\mathbb{P}}
\newcommand{\Prb}[2]{\Pr^{#1}_{#2}}  %
\newcommand{\PrA}[3]{\Prb{#1}{#2}\hspace{-0.20em}\left(#3\right)}
\newcommand{\ExSign}{\mathbb{E}}
\newcommand{\Ex}[3]{\ExSign^{#1}_{#2}\hspace{-0.17em}\left(#3\right)}

\newcommand{\Reach}[1]{\mathit{Reach}(#1)}
\newcommand{\Term}[1]{{\mathit{Term}{(#1)}}}
\newcommand{\LimInf}[2]{{\mathit{LimInf}(#1 #2)}}
\newcommand{\LimSup}[2]{{\mathit{LimSup}(#1 #2)}}
\newcommand{\CN}{{\LimInf{=}{{-}\infty}}}
\newcommand{\negCN}{{\LimInf{>}{{-}\infty}}}
\newcommand{\DI}{{\LimInf{=}{{+}\infty}}}
\newcommand{\negDI}{{\LimInf{<}{{+}\infty}}}
\newcommand{\minusDI}{{\LimSup{=}{{-}\infty}}}
\newcommand{\Mean}[1]{{\mathit{Mean}(#1 0)}}
\newcommand{\PM}{{\Mean{>}}}
\newcommand{\negPM}{{\Mean{\leq}}}
\newcommand{\Nn}{{\mathit{All}(\geq 0)}}
\newcommand{\Inf}[1]{{\mathit{Inf}(#1)}}

\newcommand{\pwR}[2]{#1[#2]} %
\newcommand{\pw}[3]{\pwR{#1}{#2}(#3)} %

\newcommand{\NP}{\textbf{NP}}
\newcommand{\coNP}{\textbf{coNP}}

\newcommand{\PTIME}{\textbf{P}}

\newcommand{\len}[1]{\mathit{len}(#1)}

\begin{document}

\title{One-Counter Stochastic Games}

\author{ T. Br\'{a}zdil\inst{1}\thanks{Supported by the Czech Science Foundation, grant No.~P202/10/1469.},
         V. Bro\v{z}ek\inst{2}\thanks{Supported by Newton International Fellowship from the Royal Society.},
         K. Etessami\inst{2}}

\institute{
Faculty of Informatics, Masaryk University \\
\email{xbrazdil@fi.muni.cz}
\and
School of Informatics, University of Edinburgh \\
\email{\{kousha,vbrozek\}@inf.ed.ac.uk}
}

\titlerunning{One-Counter Stochastic Games}   %

\maketitle

\begin{abstract}
We study the computational complexity of basic decision problems 
for {\em one-counter
simple stochastic games} (OC-SSGs), under various objectives.
OC-SSGs are 2-player turn-based stochastic games played on the
transition graph of classic one-counter automata.
We study primarily the {\em termination} objective, where
the goal of one player is to maximize
the probability of reaching counter value 0, while the other
player wishes to avoid this.
Partly motivated by the goal of understanding termination objectives, 
we also  study certain ``limit'' and ``long run
average'' reward objectives that are closely related to some  
well-studied objectives for stochastic games with rewards.
Examples of problems we address include: does player 1 have a
strategy to ensure that the counter eventually hits  0, i.e., {\em terminates},
 almost
surely, regardless of what player 2 does?  Or that the $\liminf$ (or $\limsup$)
counter value equals $\infty$ with a desired
probability?  Or that the long run average reward
is $>0$ with desired probability?
We 
show that the {\em qualitative termination problem}
for OC-SSGs
is in
$\NP{}\cap\coNP{}$, and is in P-time 
 for 1-player OC-SSGs,
or equivalently for {\em one-counter Markov Decision Processes} (OC-MDPs).
Moreover,
we show that {\em quantitative}
limit problems for OC-SSGs are in $\NP{}\cap\coNP{}$, and are
in \PTIME{}-time for 1-player OC-MDPs.
Both qualitative limit problems and qualitative termination problems 
for OC-SSGs  
are already at least as hard
as Condon's quantitative decision problem for finite-state SSGs.

\end{abstract}

\section{Introduction}
\label{sec-intro}
There is a rich literature on the computational complexity of analyzing
finite-state Markov decision processes
and stochastic games.
In recent years, there has also been some research done
on the complexity of basic analysis problems for 
classes of finitely-presented but infinite-state
stochastic models and games whose transition graphs arise from 
decidable infinite-state automata-theoretic 
models, including: context-free processes,
one-counter processes, and pushdown processes (see, e.g., \cite{EY05icalp}).
It turns out 
that such stochastic automata-theoretic 
models are intimately related to classic stochastic 
processes studied extensively in applied probability theory, such as 
(multi-type-)branching processes and (quasi-)birth-death processes (QBDs)
(see \cite{EY05icalp,EWY08,BBEKW10}).

In this paper we continue this line of work 
by studying \textbf{one-counter simple stochastic games (OC-SSGs)},
which are turn-based 2-player zero-sum 
stochastic games on transition graphs of classic 
one-counter automata.
In more detail, an OC-SSG has a finite set of control states,
which are partitioned into three types:
a set of {\em random} states, from where the next transition is chosen
according to a given probability distribution, and states 
belonging to one of two players: {\em Max} or {\em Min},
from where the respective player chooses the next transition.  
Transitions can change the state and can also change the value of 
the (unbounded) counter by at most $1$.
If there are no control states 
belonging to {\em Max} ({\em Min}, respectively),
then we call the resulting 1-player OC-SSG a {\em minimizing}
({\em maximizing}, respectively) {\em one-counter Markov decision process} 
(OC-MDP).

Fixing strategies for the two players yields a countable state Markov
chain and thus a probability space of infinite runs (trajectories).
We focus in this paper on {\em objectives} that
can be described by a (measurable) set of runs, such that 
player Max  wants to maximize,
and player Min wants to minimize,  
the probability of the objective.
The central objective studied in this paper is {\em termination}:
starting at a given control state and a given 
counter value $j > 0$, player Max (Min) wishes
to maximize (minimize) the probability of eventually hitting the counter
value $0$ (in any control state).   

Different objectives give rise to
different computational problems for OC-SSGs, aimed at
computing the value of the game, or optimal strategies, with respect
to that objective.  From general known facts about stochastic games 
(e.g., Martin's Blackwell determinacy theorem~\cite{M98}), 
it follows that the games we study are {\em determined},
meaning they have a {\em value}: we can associate with each such game a {\em value},
$\nu$, such that for every $\eps>0$, player Max has a strategy
that ensures the objective is satisfied with probability at least $\nu-\eps$
regardless of what player Min does,
and likewise player Min has a strategy
to ensure that the objective is satisfied with probability at most $\nu+\eps$.
In the case of termination objectives, the value may be {\em irrational}
even when the input data contains only rational probabilities, and this
is so
even in the purely stochastic setting without any players, i.e.,
with only {\em random} control states 
(see~\cite{EWY08}).

We can classify analysis problems for OC-SSGs into two kinds:  
{\em quantitative} analyses: 
``can the objective be achieved 
with probability at least/at most $p$'' for a given $p \in [0,1]$; or
{\em qualitative} analyses, which ask the same question but
restricted to $p \in \{0,1\}$.
We are often also interested in what kinds of strategies
(e.g., memoryless, etc.)
achieve these.

In a recent paper, \cite{BBEKW10}, we studied 
{\em one-player} OC-SSGs, i.e., OC-MDPs,  and obtained some
complexity results for them
under qualitative termination 
objectives and some quantitative limit objectives.  
The problems we studied included 
the qualitative termination problem (is the maximum probability
of termination $= 1$?) 
for {\em maximizing} OC-MDPs.
We showed that this problem is decidable in \PTIME{}-time.   
However, we left open the complexity of the same problem 
for {\em minimizing} OC-MDPs
(is the minimum probability of termination $< 1$?).
One of the main results of this paper is the following, which 
in particular resolves this open question:

\begin{theorem} {\bf (Qualitative termination)} \label{thm:NT-val}
Given a OC-SSG, $\G$, with the objective of termination,
and given an initial control state $s$ and initial counter
value $j > 0$,  
deciding whether the value of the game is equal to 1 
is in $\NP{}\cap\coNP{}$.
Furthermore, the same problem  is in \PTIME{}-time for 1-player OC-SSGs, i.e., for 
both maximizing and minimizing OC-MDPs.
\end{theorem}

\noindent Improving on this $\NP{}\cap\coNP{}$ upper bound for
the qualitative termination problem for OC-SSGs would require 
a breakthrough: 
we show that deciding whether the value 
of an OC-SSG termination game is equal to $1$
is already at least as hard as Condon's~\cite{C92} {\em quantitative} 
reachability problem for finite-state simple stochastic games 
(Corollary~\ref{cor:hard}).  We do not know a reduction in the
other direction.
We furthermore show that if the value is $1$ for a OC-SSG termination game, 
then Max has a simple kind of 
optimal strategy 
(memoryless, counter-oblivious, and pure)
that ensures termination
with probability $1$, regardless of Min's strategy. 
Similarly, if the value is less than $1$,
we show Min has a simple strategy (using finite memory, linearly bounded in the number
of control states) 
that ensures the probability of termination
is  $< 1-\delta$ for some positive $\delta>0$, regardless of what Max does. 
We show that such strategies for both players are 
computable in non-deterministic polynomial time
for OC-SSGs, and in deterministic P-time for (both maximizing
and minimizing) 1-player OC-MDPs.
We also observe that the analogous problem of deciding whether the value of a OC-SSG termination
game is $0$ is in \PTIME{}, which
follows easily by reduction to non-probabilistic games.

OC-SSGs can 
be viewed as stochastic game extensions
of Quasi-Birth-Death Processes (QBDs) (see~\cite{EWY08,BBEKW10}).   
QBDs are a heavily studied model in queuing theory
and performance evaluation (the counter keeps track of
the number of jobs in a queue).
It is very natural to consider controlled and game extensions
of such queuing theoretic models, thus allowing
for adversarial
modeling of queues with unknown (non-deterministic) environments
or with other unknown aspects modeled non-deterministically.
OC-SSGs with termination objectives also subsume ``solvency games'', a recently
studied class of MDPs motivated by modeling of a risk-averse investment scenario~\cite{BKSV08}. 

Due to the presence of an unbounded counter, an
OC-SSG, $\G$, formally describes a stochastic game
with a countably-infinite state space:  a ``configuration'' or ``state'' of the 
underlying stochastic game consists of a pair $(s,j)$, where $s$ is a control 
state of $\G$ and $j$ is the current counter value.
However, it is easy to see that we can equivalently view $\G$ as a
finite-state
\textbf{simple stochastic game (SSG)}, $\H$, with \textbf{rewards} as follows:
$\H$ is played on the finite-state transition graph
obtained from that of $\G$ by simply ignoring
the counter values.
Instead, every transition $t$ of $\H$ is assigned a \emph{reward}, 
$r(t) \in \{-1,0,1\}$,
corresponding to the effect that the transition $t$ would have on the counter
in $\G$.
Furthermore, when emulating an OC-SSG using
rewards,
we can easily place rewards on states 
rather than on 
transitions, by adding suitable auxiliary control states.
Thus, w.l.o.g., we can assume that OC-SSGs are presented
as equivalent finite-state SSGs with a reward, $r(s) \in \{-1,0,1\}$ 
labeling each state $s$.
A {\em run} of $\H$, $w$, is an infinite sequence of states
that is permitted by the transition structure, and 
we denote the $i$-th state along the run $w$ by $w(i)$.  
The termination objective for $\G$, when the initial counter value is $j > 0$, 
can now be rephrased as the
following equivalent objective for $\H$: 
\[
\Term{j} \coloneqq 
\{\, w \mid w \text{ is a run of } \H \text{ such that there exists }
m > 0 \text{ such that } \textstyle\sum\nolimits_{i=0}^m r(w(i)) = -j\: \}\ .
\]
An important step toward our 
proof of Theorem~\ref{thm:NT-val} and related results,
is to establish links between this termination objective and the following 
limit 
objectives,
which are of independent interest.
For  $z \in \{-\infty, \infty\}$, and a comparison operator
$\Delta \in \{ > , <,  =\}$, 
consider the following objective:
\[
\LimInf{\mathop{\Delta}}{z} \coloneqq
\{\, w \mid w\text{ is a run of } \H \text{ such that }
\liminf_{n\to\infty} \textstyle\sum\nolimits_{i=0}^n r(w(i)) \mathop{\Delta} z \:\}
\ .
\]
We will show that if $j$ is large enough (larger than the number of control 
states), then
the game value with respect to objective $\Term{j}$ and
the game value with respect to $\CN$ are either
both equal to $1$, or are both less than $1$ (Lemma~\ref{lem:NT-MD}).
We could also consider the ``$\sup$'' variant of these objectives, such as
$\minusDI$, but these are
redundant.  For example, by negating the sign of rewards,
$\minusDI$ is ``equivalent'' 
to $\DI$.
Indeed, the only limit objectives
we need to consider for SSGs are
$\CN$ and $\DI$,
because the others are either the same objectives
considered from the other player's points of view, or 
are vacuous, such as $\LimInf{>}{{+}\infty}$.
For both limit objectives, $\CN$ and $\DI$, 
we shall see that the value of the respective
SSGs is always rational (Proposition~\ref{prop:val-rational}).
We shall also show that
the objective $\DI$ is essentially equivalent to the following ``mean payoff'' 
objective (Lemma~\ref{lem:PM-DI-eq}):
\[
\PM \coloneqq
\{\, w \mid \text{w is a run of } \H \text{ such that }  
\liminf_{n\to\infty} \textstyle\sum\nolimits_{i=0}^{n-1} r(w(i))/n>0 \: \}\ .
\]
This ``intuitively obvious equivalence'' 
is not so easy to prove. 
(Note also that $\CN$ is certainly not equivalent
 to $\negPM$.)
We establish the equivalence by a combination of new methods
and by using
recent results by Gimbert, Horn and Zielonka~\cite{GH-SODA10,GZ09}.  
Mean payoff objectives are of course very heavily studied 
for stochastic games and for MDPs (see~\cite{Puterman94}).
However, there is a subtle but important difference here:
mean payoff objectives are  typically
formulated via {\em expected payoffs}: the Max player
wishes to maximize the {\em expected} mean payoff, while the
Min player wishes to minimize this.
Instead, in the above $\PM$ objective
we wish to maximize (minimize) the {\em probability} that
the mean payoff is $> 0$.   These require new algorithms.
Our main result about such limit objectives is the following:

\begin{theorem}\label{thm:rewards}
For both limit objectives, $O\in\{\CN,\DI\}$,
given a finite-state SSG, $\G$, with rewards,
and given a rational probability threshold, $p,\, 0\leq p\leq 1$,
deciding whether the value of $\G$
with objective $O$ 
is ${>} p$ (or ${\geq} p$) is in $\NP{}\cap\coNP{}$.
If $\G$ is a 1-player SSG (i.e., a maximizing or minimizing 
MDP), then the game value 
can be computed in P-time.
\end{theorem}
\noindent
Although our upper bounds for both these objectives look the same,
their proofs are quite different.
We show that both players have pure and
memoryless optimal strategies
in these games (Proposition~\ref{prop:long-run-MD-det}),
which can be  computed in P-time for 1-player (Max or Min) MDPs.
Furthermore, we show that even deciding whether the value
of these games is either $1$ or $0$, given input for which 
one of these two is promised to be the case, is already at least as hard as
Condon's~\cite{C92} {\em quantitative} reachability
problem for finite-state simple stochastic games
(Proposition~\ref{prop:nb-hard}).
Thus, even any non-trivial {\em approximation} of the value of SSGs with such limit objectives is not easier than Condon's problem.

We already considered in \cite{BBEKW10} the
problem of maximizing the probability
of $\CN$ in a OC-MDP.  
There we showed that the maximum probability can be computed in P-time.
However, again, we did not resolve 
the complementary problem of minimizing the probability of
$\CN$ in a OC-MDP. 
Thus we could not address two-player OC-SSGs with either of these objectives,
and we left these as key open problems, which we resolve here.
An important distinction between {\em maximizing}
and {\em minimizing} 
the probability of objective $\CN$
is that maximizing this objective satisfies a {\em submixing}
property defined by Gimbert \cite{Gimbert-STACS07}, which he showed
implies the existence of optimal memoryless 
strategies, whereas minimizing the objective is not submixing, and
thus we require new methods to tackle it, which we develop in this paper.

Finally, we mention that
one can also consider OC-SSGs with the objective of
terminating in a {\em selected} 
subset of states, $F$. 
Such objectives were considered for OC-MDPs in \cite{BBEKW10}.
Using our termination results in this paper, we can 
also show
that given an OC-SSG it is decidable  (in double exponential time)
whether Max can achieve a termination probability 1
in a selected subset of states, $F$.
The computational complexity of selective termination is higher
than for non-selective termination: 
PSPACE-hardness holds already for OC-MDPs without Min
(\cite{BBEKW10}).
Due to space limitations, we omit results about selective termination 
from this
conference paper, and will
include them in the journal version of this paper.

\vspace{-1em}

\paragraph{Related work.}
As mentioned earlier, we initiated the study of some classes 
of 1-player OC-SSGs 
(i.e., OC-MDPs) in a recent paper \cite{BBEKW10}.
The reader will find extensive references to earlier
related literature in~\cite{BBEKW10}. 
No earlier work considered OC-SSGs explicitly,
but as we have highlighted already there are close connections 
between OC-SSGs and finite-state stochastic games with certain 
interesting limiting average reward objectives.
One-counter automata with a non-negative counter are equivalent 
to 
pushdown automata restricted to a 1-letter stack alphabet
(see \cite{EWY08}), and  
thus OC-SSGs with the termination objective form a subclass of
pushdown stochastic games, or equivalently, Recursive 
simple stochastic games (RSSGs).
These more general stochastic games were introduced 
and studied in~\cite{EY05icalp}, where it is shown that
many interesting computational 
problems for the general RSSG and RMDP models are undecidable,
including generalizations of qualitative termination problems
for RMDPs.
It was also established in
\cite{EY05icalp} that for  stochastic context-free games (1-exit RSSGs),
which correspond to pushdown stochastic games with only one state,
both qualitative and quantitative termination problems are decidable,
and in fact qualitative termination problems are decidable in $\NP{}\cap\coNP {}$
(\cite{EY06stacs}).
Solving termination objectives is a key ingredient for many more
general analyses and model checking problems for stochastic games.
OC-SSGs form another natural subclass of RSSGs,
which is incompatible with stochastic context-free games.
Specifically, for OC-SSGs with the termination objective,
the number of stack symbols, rather than the number of control states, 
of a pushdown stochastic game is being restricted to $1$.
As we show in this paper, this restriction again yields decidability
of the qualitative termination problem. 
However, the decidability of the quantitative termination problem 
for OC-SSGs
remains an  open problem (see below).

\vspace{-1em}

\paragraph{Open problems.}
Our results complete part of the picture  for decidability
and complexity of several problems for OC-SSGs.
However, our results also leave many open questions.
The most important open question for OC-SSGs is whether the {\em quantitative} 
termination problem, even for OC-MDPs, is decidable.
Specifically, we do not know 
whether the following is decidable:  given a OC-MDP, and a
rational probability $p \in (0,1)$,  
decide whether the maximum probability of
termination is ${>} p$  (or ${\geq} p$).
Substantial new obstacles arise for deciding
this.  In particular, we know that an optimal strategy
may in general need to use different actions at the same
control state for arbitrarily large counter values 
(so strategies cannot ignore the value of the counter, 
 even for arbitrarily large values), and this holds
already for the extremely simple case of solvency games
\cite[Theorem~3.7]{BKSV08}.

\vspace{-1em}

\paragraph{Outline of paper.}
We fix notation and key definitions in Section~\ref{sec-defs}.
In Section~\ref{sec:noboundary}, we prove Theorem~\ref{thm:rewards}.
Building on Section~\ref{sec:noboundary},
we prove Theorem~\ref{thm:NT-val} in Section~\ref{sec:boundary}.
Many proofs are in the appendix.

\section{Preliminaries}
\label{sec-defs}  

\begin{definition}
A {\bf simple stochastic game (SSG)}
is given by a finite, or countably infinite directed graph, $(V,\btran{})$,
where $V$ is the set of {\em vertices} (which we also call {\em states}),
and $\btran{}$ is the edge (also called {\em transition}) relation,
together with
a partition $(V_\top,V_\bot,V_P)$ of $V$,  
as well as 
a \emph{probability assignment}, $\Prob$, which
to each $v \in V_P$ assigns a rational probability distribution on
its set of outgoing edges.
States in $V_P$ are called \emph{random}, states in
$V_\top$ belong to player Max, and states in
$V_\bot$ belong to player Min.
We assume that
for every $v\in V$ there is at least one $u\in V$
such that $v\btran{}u$.
We often write $v\btran{x}u$ instead of $\Prob(v\btran{}u)=x$.
If $V_\bot=\emptyset$ we call $\G$
a \emph{maximizing} {\bf Markov decision process (MDP)}.
If $V_\top=\emptyset$ we call it a \emph{minimizing} MDP.
If $V_\bot=V_\top=\emptyset$ then we call $\G$ a
{\bf Markov chain}.
A SSG (or a MDP or a Markov chain) can be equipped with a reward
function, $r$,
which assigns to each state, $v\in V$,
a number $r(v)\in\{-1,0,1\}$.\footnote{Rewards 
can generally be arbitrary rational values, but for 
this paper we confine ourselves to rewards in $\{-1,0,1\}$.} 
Similarly, rewards can be assigned to transitions.
\end{definition}

For a {\em path},
$w=w(0) w(1) \cdots w(n-1)$, of states
in a graph,
we use $\len{w}=n$ to denote the length of $w$. 
A \emph{run} in a SSG, $\G$, is an infinite path in the underlying directed graph.
The set of all runs in $\G$ is denoted by $\run_{\G}$, and
the set of all runs starting with a finite path $w$ is $\run_{\G}(w)$.
These sets generate the standard Borel algebra on $\run_{\G}$.

A \emph{strategy} for player Max is a function,
$\sigma$, which to each \emph{history} $w \in  V^+$
ending in some $v\in V_\top$, assigns a
probability distribution on the set of outgoing transitions of~$v$.
We say that a strategy
$\sigma$ is \emph{memoryless} if $\sigma(w)$ depends only on the
last state, $v$, and \emph{pure} if $\sigma(w)$
assigns probability 1 to some transition, for each history $w$. 
When $\sigma$ is pure, we write $\sigma(w)=v'$ instead
of $\sigma(w)(v,v')=1$.
Strategies for player Min are defined similarly, just by substituting
$V_\top$ with $V_\bot$.

Assume we fix
a starting state $s$, and
a pair of strategies: $\sigma$ for player Max,
and $\pi$ for Min in a SSG, $\G$.
There is a unique probabilistic measure, $\Prb{\sigma,\pi}{s}$,
on the Borel space of runs $\run_\G$,
satisfying for all finite paths $w$ starting in $s$:
$\PrA{\sigma,\pi}{s}{\run_{\G}(w)} = \prod_{i{=}1}^{\len{w}-1} x_i$ where
$x_i,\ 1 \leq i < \len{w}$ are defined by requiring that
(a) if $w(i{-}1) \in V_P$ then $w(i{-}1)\btran{x_i}w(i)$; and
(b) if $w(i{-}1) \in V_\top$ then $\sigma(w(0)\cdots w(i{-}1))$ assigns $x_i$
to the transition $w(i{-}1)\btran{}w(i)$; and
(c) if $w(i{-}1) \in V_\bot$ then $\pi(w(0)\cdots w(i{-}1))$ assigns $x_i$
to the transition $w(i{-}1)\btran{}w(i)$.
In particular, $\PrA{\sigma,\pi}{s}{\run_{\G}(s)} = 1$.
In cases where $\G$ is a maximizing MDP, a minimizing MDP,
or a Markov chain, we denote this probability measure by
$\Prb{\sigma}{s}$,
$\Prb{\pi}{s}$, or
$\Prb{}{s}$, respectively.
See, e.g., \cite[p.~30]{Puterman94}, for the existence and uniqueness
of the measure $\Prb{\sigma}{s}$ in the case of MDPs. 
It is straightforward then to establish existence and uniqueness of 
$\Prb{\sigma,\pi}{s}$ for SSGs, by considering pairs of strategies to be
one strategy.

In this paper, an \emph{objective} for a stochastic game is
given by a measurable set of runs.
An objective, $O$, is called a \emph{tail} objective
if for all runs $w$ and all suffixes $w'$ of $w$, we have 
$w'\in O \iff w\in O$.

Assume we have fixed a SSG, an objective, $O$, and a starting state, $s$.
We define the \emph{value of $\G$ in $s$} as
\(
\Val^{O}(s)  \coloneqq  \sup_{\sigma} \inf_{\pi} \PrA{\sigma,\pi}{s}{O}
\).
It follows from Martin's Blackwell determinacy theorem \cite{M98} that
these games are {\em determined}, meaning
$\Val^{O}(s) =  \inf_{\pi} \sup_{\sigma} \PrA{\sigma,\pi}{s}{O}$.
A strategy $\sigma$ for Max is \emph{optimal in $s$}
if $\PrA{\sigma,\pi}{s}{O} \geq \Val^{O}(s)$ for every $\pi$.
Similarly a strategy $\pi$ for Min is \emph{optimal in $s$}
if $\PrA{\sigma,\pi}{s}{O} \leq \Val^{O}(s)$ for every $\sigma$.
A strategy is called {\em optimal} if it is optimal in every state.

An important objective for us is \emph{reachability}. Given a set
$T \subseteq V$, we define the objective
$\Reach{T} \coloneqq \{w \in \run_{\G} \mid \exists i \geq 0 : w(i) \in T \}$.
The following fact is well known:

\begin{fact}
\label{fact:reach}
(See, e.g., \cite{Puterman94,C92,CY98}.)
For both maximizing and minimizing finite-state MDPs with reachability objectives,
pure memoryless optimal strategies
exist and can be computed, together with the optimal value, in polynomial time.
\end{fact}

\section{Limit objectives}\label{sec:noboundary}

All MDPs and SSGs in this section have
finitely many states.
Rewards are assigned to states, not to transitions.
The main goal of this section is to prove Theorem~\ref{thm:rewards}.
We start by proving that both players have optimal
pure and memoryless strategies for objectives $\CN$, $\DI$, and $\PM$.
The following is a corollary of a result by Gimbert and Zielonka, 
which allows us to concentrate on MDPs instead of SSGs:

\begin{fact}
\label{fact:mdp-games}
(See~\cite[Theorem 2]{GZ09}.)
Fix any objective, $O$, and suppose that in every maximizing
and minimizing MDP with objective $O$, the unique player has a pure memoryless 
optimal strategy.
Then in all SSGs with objective $O$, 
both players have optimal pure and memoryless strategies.
\end{fact}

\smallskip

Note that the probability of $\CN$ is minimized iff
the probability of $\negCN$ is maximized, similarly with
$\DI$ vs.\ $\negDI$, and $\PM$ vs.\ $\negPM$.

\smallskip
\begin{fact}
\label{fact:md-qual}
(See~\cite[Theorem 4.5]{GH-SODA10}.)
Let $O$ be a tail objective.
Assume that for every maximizing MDP and for every state, $s$, with $\Val^O(s)=1$,
there is an optimal pure memoryless strategy starting in $s$.
Then for all $s$ there is an optimal pure memoryless strategy 
starting in $s$,
without restricting $\Val^O(s)$.
\end{fact}

\begin{proposition}\label{prop:long-run-MD-det}
For every SSG, considered with any of the objectives
$\CN$, $\DI$, or $\PM$,
both players Max and Min
have optimal pure memoryless strategies.
\end{proposition}

\begin{proof} {\em (Sketch.)}
Using Fact~\ref{fact:mdp-games} we consider only maximizing MDPs,
and prove the proposition for the objectives listed and
their complements.
Note that since all these objectives are tail, a play under an optimal strategy,
starting from a state with value $1$, cannot visit a state with value $<1$.
By Fact~\ref{fact:md-qual} we may thus safely assume that the value
is $1$ in all states.
We discuss different groups of objectives:

\smallskip

\noindent {$\CN$, $\negDI$, $\negPM$, $\PM$:}
The first three
(with $\CN$ also handled explicitly in~\cite{BBEKW10})
are {\em tail} objectives
and are also {\em submixing} (see \cite{Gimbert-STACS07}).
Therefore, Theorem~1 of~\cite{Gimbert-STACS07} immediately yields the 
desired result.
$\PM$ can be equivalently rephrased via a submixing $\limsup$ variant.
See Section~\ref{sec:sumbix} in the appendix for details.

\smallskip

\noindent {$\DI$:}
is a tail objective,
so there is always a pure optimal strategy, $\tau$,
by~\cite[Theorem 3.1]{GH-SODA10}.
Note that $\DI$ is {\em not submixing}, so 
Theorem~1 of \cite{Gimbert-STACS07} does not apply.
In the following we proceed in two steps:
we start with $\tau$ and convert it
to a finite-memory strategy%
\footnote{
A finite-memory strategy
is specified by a finite state automaton, $\A$, over the alphabet $V$. Given $w\in V^+$, the value
$\sigma(w)$ is determined by the state of $\A$ after reading $w$.},
$\sigma$.
Finally, we reduce the use of memory to get a memoryless strategy.

First, we obtain a finite-memory optimal strategy, 
starting in some state, $s$. For a run $w\in \run_{\G}(s)$ and $i\geq 0$,
we denote by $\pw{r}{i}{w}$ the accumulated reward
$\sum_{j=0}^i r(w(j))$ up to step $i$.
Observe that because $\tau$ is optimal
there is some $m>0$ and a (measurable) set of runs $A\subseteq \run_{\G}(s)$,
such that $\PrA\tau{s}{A}\geq \frac{1}{2}$, and for all $w\in A$ we have that
the accumulated reward along $w$ 
never reaches $-m$ (i.e.\ $\inf_{i\geq 0} \pw{r}{i}{w}>-m$).
Since for almost all runs of $A$ we have 
$\lim_{i\to\infty} \pw{r}{i}{w}=\infty$, there is some $n>0$ and a 
set $B\subseteq A$
such that $\PrA\tau{s}{B\mid A}\geq\frac{1}{2}$ (and hence, $\PrA\tau{s}{B}\geq \frac{1}{4}$), and
for all $w\in B$ we have that the accumulated reward along $w$ reaches $4m$ before the $n$-th step.
Thus with probability at least $\frac{1}{4}$, a run $w\in \run_{\G}(s)$
satisfies $\inf_{i\geq 0} \pw{r}{i}{w}>-m$ and
$\max_{0\leq i\leq n} \pw{r}{i}{w}\geq 4m$.

We denote by $T_s(w)$ the {\em stopping time} over
$\run_{\G}(s)$ which for every $w\in \run_{\G}(s)$ returns the
least number $i\geq 0$ such that either $\pw{r}{i}{w}\not\in (-m,4m)$, or
$i=n$. Observe that the expected accumulated reward at the stopping time
$T_s$ is at least $\frac{1}{4}\cdot 4m + \frac{3}{4} (-m)=\frac{m}{4}>0$.
Let us define a new strategy $\sigma$ as follows. Starting in a state
$s\in V$, the strategy $\sigma$ chooses the same transitions as $\tau$
started in $s$, up to the stopping time $T_s$.  Once the stopping time
is reached, say in a state $v$, the strategy $\sigma$ erases its memory
and behaves like $\tau$ started anew in $v$. Subsequently, $\sigma$
follows the behavior of $\tau$ up to the stopping time $T_v$. Once the
stopping time $T_v$ is reached, say in a state $u$, $\sigma$ erases
its memory and starts to behave as $\tau$ started anew in $u$, and so on.
Observe that the strategy $\sigma$ uses only finite memory because each
stopping time $T_s$ is bounded for every state $s$.
Because $\tau$ is pure, so is $\sigma$.

Now we argue that $\sigma$ is optimal.  Intuitively, this is because,
on average, the accumulated reward strictly increases between resets
of the memory of $\sigma$. To formally argue that this implies that the
accumulated reward increases indefinitely, we employ the theory of random
walks on $\mathbb{Z}$ and sums of i.i.d.\ random variables (see, e.g.,
Chapter~8 of~\cite{Chung01}). Essentially, we define a set of random walks,
one for each state $s$, capturing the sequence of changes to the accumulated
reward between each reset in $s$ and the next reset (in any state).
We can then apply random walk results, e.g., from \cite[Chapter~8]{Chung01}, to conclude that these walks
diverge to $\infty$ almost surely.
For details see Lemmas~\ref{lem:fin-opt} and~\ref{lem:MC-classify}
in the appendix.

Taking the product of the finite-memory strategy $\sigma$ and $\G$
yields a finite-state Markov chain.
By analyzing its bottom strongly connected components
we can eliminate the use of memory, and obtain a pure and memoryless optimal
strategy, see Lemma~\ref{lem:DI-MD} in the appendix.

\smallskip

\noindent {$\negCN$:}
Like $\DI$, the objective $\negCN$ is tail, but not submixing.
Thus there is always a pure optimal strategy, $\tau$, for $\negCN$,
by~\cite[Theorem 3.1]{GH-SODA10}, but
Theorem~1 of \cite{Gimbert-STACS07} does not apply.
We will prove Proposition~\ref{prop:long-run-MD-det} for
$\negCN$ using the results for $\DI$, and also a new objective,
\(
\Nn \coloneqq 
\{ w\in \run_{\G} \mid \forall n\geq 0:\textstyle\sum\nolimits_{j=0}^n r(w(j))\geq 0\}
\).
Let $W_\infty$ and $W_+$ denote the
sets of states $s$ such that $\Val^{\DI}(s)=1$,
and $\Val^\Nn(s)=1$, respectively.
Then, as we prove in the appendix, Lemma~\ref{lem:CN-di-nn}, for every state, $s$, with $\Val^\negCN(s)=1$:
\begin{equation}
\label{eq:reduce_reach}
\exists \sigma:
\PrA\sigma{s}{\Reach{W_\infty \cup W_+}}=1
\end{equation}
Moreover, we prove that whenever $\Val^\Nn(s) = 1$ then Max has
a pure and memoryless strategy $\sigma_+$ which is optimal in $s$ for $\Nn$.
Indeed, observe that player Max achieves $\Nn$ with probability $1$ iff {\em all} runs 
satisfy it. So we may consider
the MDP $\G$ as a 2-player non-stochastic game, where random nodes 
are now treated as player Min's. 
In this case, Theorem~12 of~\cite{VASS} guarantees the existence of
the promised strategy $\sigma_+$.
The proof is now finished by observing that, by Fact~\ref{fact:reach}, there is a
pure and memoryless strategy
$\sigma$ maximizing the probability of reaching $W_\infty\cup W_+$.
The resulting pure and memoryless strategy, optimal for $\negCN$, can be obtained by ``stitching''
$\sigma$ together with the respective optimal strategies for $\DI$ and $\Nn$.
\qed
\end{proof}

\noindent
The following simple lemma is proved in the appendix.
\newcommand\lemmabscctext{%
Let $\M$ be a finite, strongly connected (irreducible) Markov chain,
and $O$ be a tail objective. Then there is $x\in\{0,1\}$
such that $\PrA{}{s}{O}=x$ for all states $s$.%
}
\begin{lemma}\label{lem:01bscc}
\lemmabscctext{}
\end{lemma}

\noindent A corollary of the previous proposition and lemma is the following:
\begin{proposition}\label{prop:val-rational}
Let $O\in \{\CN, \DI, \PM\}$. Then
in every SSG,
and for all states, $s$,
$\Val^O(s)$ is rational,
with a polynomial length
binary encoding.
\end{proposition}
\begin{proof}
By Proposition~\ref{prop:long-run-MD-det},
there are memoryless optimal strategies: $\sigma$
for Max, and $\pi$ for Min. Fixing them induces a Markov
chain on the states of $\G$. By Lemma~\ref{lem:01bscc},
in every fixed bottom strongly connected component (BSCC), $C$,
of this finite-state Markov chain,
all states $v \in C$ have the same value, $x_C$, which is either $0$ or $1$.
Denote by $W$ the union of all BSCCs, $C$, with $x_C=1$.
By optimality of $\sigma$ and $\pi$,
$\Val^{O}(s) = \PrA{\sigma,\pi}{s}{\Reach{W}}$ for
every $s\in V$. By, e.g., \cite[Section~3]{CY98}, this probability
is rational, with polynomial length bit encoding, since
reaching $W$ is a regular event, and every Markov chain
is a special case of a MDP.
\qed
\end{proof}

\vspace{-1em}

\paragraph{Proof of Theorem~\ref{thm:rewards}.}
We will need a couple of preliminary lemmas:

\begin{lemma}\label{lem:PM-DI-eq}
Let $\G$ be a MDP with rewards, and $s$ a state of $\G$.
Then for every memoryless strategy $\sigma$:
\[
\PrA\sigma{s}{\PM}=
\PrA\sigma{s}{\DI}
\]
In particular, both objectives are equivalent with respect to 
both the value
and optimal strategies.
\end{lemma}
\begin{proof} {\em (Sketch.)}
The inequality $\leq$ is true for all strategies, since $\PM \subseteq \DI$.
In the other direction, the property that $\sigma$ is memoryless is needed,
so that fixing $\sigma$ yields a Markov chain on the states of $\G$.
In this Markov chain, by Lemma~\ref{lem:01bscc}, for every BSCC, $C$,
there are $x_C\leq y_C\in\{0,1\}$, such that
\(
\PrA\sigma{s}{\PM \mid \Reach{C}}=x_C
\), and
\(
\PrA\sigma{s}{\DI \mid \Reach{C}}=y_C
\).
By random walk arguments, considering the rewards accumulated between subsequent
visits to a fixed state in $C$,
we can prove that $y_C=1 \implies x_C=1$,
see Lemma~\ref{lem:PM-DI-scc} in the appendix.
Proposition~\ref{prop:long-run-MD-det} finishes the proof.
\qed
\end{proof}

\begin{lemma}\label{lem:reach-mdp-qual}
For an objective  $O\in \{\CN, \negCN, \DI, \negDI\}$,
and a maximizing MDP, $\G$,
denote by $W$ the set of all $s\in V$ satisfying $\Val^O(s)=1$.
Then
$\Val^O(s)=\Val^{\Reach{W}}(s)$ for every state $s$.
\end{lemma}

\begin{proof}
Proposition~\ref{prop:long-run-MD-det} gives us
a memoryless optimal strategy, $\sigma$.
By fixing it, we obtain a Markov chain on states of $\G$.
We denote by $W'$ the union of all BSCCs of this
Markov chain,
in which at least one state has a positive value.
By Lemma~\ref{lem:01bscc}, all states from $W'$ have, in fact, value $1$.
Since $W'\subseteq W$, and $\sigma$ is optimal, we get
\[
\Val^O(s)
=\PrA\sigma{s}{O}
=\PrA\sigma{s}{\Reach{W'}}
\leq\PrA\sigma{s}{\Reach{W}}
\leq\Val^{\Reach{W}}(s)
\] for every state $s$.
Because $O$ is a tail objective, we easily obtain
$\Val^O(s)\geq\Val^{\Reach{W}}(s)$.
\qed
\end{proof}

To prove Theorem~\ref{thm:rewards}, we start with the MDP case.
By Proposition~\ref{prop:long-run-MD-det},
pure memoryless
strategies are sufficient for optimizing the probability
of all the objectives considered in this theorem, 
so we can restrict ourselves to 
such strategies for 
this proof.
Given an objective $O$, we will write $W^{O}$ to denote the set of
states $s$ with $\Val^{O}(s)=1$.
As $\G$ is a MDP, optimal strategies for {\em reaching} any state in 
$W^O$ can be computed in 
polynomial time, by Fact~\ref{fact:reach}.
If $O$ is any of the objectives mentioned in the statement of 
Lemma~\ref{lem:reach-mdp-qual}, then 
by that Lemma, in order to compute optimal strategies and values for objective
$O$, 
it suffices to compute the set $W^{O}$ and optimal strategies 
for the objective $O$ in
states in $W^{O}$. The resulting optimal strategy ``stitches''
these and the optimal strategy for reaching $W^O$.

\begin{proposition}\label{prop:max-comp}
For every MDP, $\G$, and an objective $O=\CN$, $\DI$, or $\PM$,
the problem whether $s\in W^{O}$ is decidable in \PTIME{}-time.
If $s\in W^{O}$, then a strategy optimal in $s$ is computable in \PTIME{}-time.
\end{proposition}

\begin{proof} {\em (Sketch.)}
From Lemma~\ref{lem:PM-DI-eq}
we know that $\DI$ is equivalent to $\PM$, and thus we only have to consider
$O=\CN$ and $O=\PM$.
For a uniform presentation, we assume that $\G$ is a maximizing MDP, and consider
two cases: $O=\PM$,  and $\negCN$.
The remaining cases were solved in~\cite{BBEKW10} -- 
Theorem~3.1 there solves the case $O=\CN$, and Section~3.3 solves $O=\negPM$.

\smallskip

\noindent {$O = \PM $}:
We design an algorithm to decide whether 
$\max_{\sigma}\PrA{\sigma}{s}{\PM}=1$,
using the existing polynomial time algorithm,
based on linear programming, for maximizing the {\em expected} mean payoff
and computing optimal strategies for it 
(see, e.g.,~\cite{Puterman94}).
Note that, as shown in the appendix (Lemma~\ref{lem:MD-lim-def}),
it does not matter whether $\liminf$ or $\limsup$
is used in the definition of $\PM$.
Under a memoryless strategy $\sigma$, almost all runs in $\G$ reach one of the 
bottom strongly 
connected components (BSCCs).
Almost all runs initiated in some BSCC, $C$,
visit all states of $C$ infinitely often,
and it follows from standard Markov chain theory (e.g., \cite{Norris98}) 
that almost all runs in $C$
have the same mean payoff, which equals the expected mean payoff
for the Markov chain induced by $C$.

\begin{procedure}[h]
\caption{MP($s$)}\label{proc:mp}
\dontprintsemicolon
\KwData{A state $s$.}
\KwResult{Decide $\Val^\PM(s)\stackrel{\textrm{?}}{=}1$. If yes, return a strategy $\sigma$ with
$\PrA\sigma{s}{\PM}=1$.}

\Repeat{$s$ is cut off \nllabel{proc:mp-until}}{
Compute a strategy $\sigma_{mp}$ maximizing the expected mean payoff.
\;
\nllabel{proc:mp-exmax}

\lIf{$\Ex{\sigma_{mp}}{s}{\text{mean payoff}}\leq 0$}{\Return{\DataSty{No}}}
\;
\nllabel{proc:mp-no}

Fix $\sigma_{mp}$ to get a Markov chain on $\G$. Find a BSCC, $C$, with mean payoff almost surely positive.\\
\;
\nllabel{proc:mp-bscc}

Compute a strategy $\sigma_C$ maximizing the probability
of $\Reach{C}$.
\;
\nllabel{proc:mp-maxC}

\ForEach{$v$ with $\PrA{\sigma_C}{v}{\Reach{C}}=1$}%
{
Remove state $v$.
\;
\nllabel{proc:mp-cut}

\lIf{$v\in C$}%
{$\sigma(v) \leftarrow \sigma_{mp}(v)$}
\lElse%
{$\sigma(v) \leftarrow \sigma_C(v)$}
}
}%

\Return{$(\DataSty{Yes},\sigma)$}
\;
\nllabel{proc:mp-yes}
\end{procedure}

The algorithm is given here as Procedure~\FuncSty{MP}($s$).  Both
step~\ref{proc:mp-exmax}, as well as verifying the condition from
step~\ref{proc:mp-bscc}, can be done in
P-time, because, as observed above, this is equivalent to verifying
that the expected mean payoff in $C$ is positive, which
can be done in P-time (see \cite[Theorem~9.3.8]{Puterman94}).
Step~\ref{proc:mp-maxC} can be done in P-time by Fact~\ref{fact:reach}.
To obtain a formally correct MDP, we introduce a new state $z$ with a 
self-loop,
and after the removal of any state $v$ in step~\ref{proc:mp-cut}
of the for loop, 
we redirect all stochastic
transitions leading to $v$ to this new state $z$,
and eliminate all other transitions into $v$. The
reward of the new state $z$ is set to $0$. This will not affect the sign of
subsequent optimal expected mean payoffs starting from $s$, unless $s$
has been already removed.  Thus, the algorithm can be implemented so
that each iteration of the repeat-loop takes P-time, and so
the algorithm terminates in P-time, since in each
iteration at least one state must be removed.
If the algorithm outputs $(\DataSty{Yes},\sigma)$ then clearly
$\PrA\sigma{s}{\PM}=1$. On the other hand, by an easy
induction on the number of iterations of the repeat-loop
one can prove that if $\Val^\PM(s)=1$ then the following is an invariant
of line~\ref{proc:mp-until}: either $s$ has been removed, or
the maximal expected mean payoff starting in $s$ is positive.
In particular, the algorithm cannot output $\DataSty{No}$.
Thus we have completed the case when $O = \PM$.

\smallskip

\noindent {$O=\negCN$:}
Recall first the auxiliary objective
\(
\Nn \coloneqq 
\{ w\in \run_{\G} \mid \forall n\geq 0:\textstyle\sum\nolimits_{j=0}^n r(w(j))\geq 0\}
\)
from the proof of Proposition~\ref{prop:long-run-MD-det},
and also the sets
$W_\infty = \{ v \mid \Val^{\DI}(v)=1\}$, and
$W_+ = \{ v \mid \Val^\Nn(v)=1\}$.
Note that $W_\infty=W^\PM$, by Lemma~\ref{lem:PM-DI-eq}.
Finally, recall from the equation (\ref{eq:reduce_reach})
in the proof of Proposition~\ref{prop:long-run-MD-det},
that the probability of $\negCN$ is maximized
by almost surely reaching $W_\infty\cup W_+$
and then satisfying $\Nn$ or $\DI$.
We note that the strategy $\sigma_+$, optimal for $\Nn$, from the proof of
Proposition~\ref{prop:long-run-MD-det}, can be computed in polynomial time
by~\cite[Theorem~12]{VASS}.
The results on $\PM$ and Fact~\ref{fact:reach} conclude the proof.
\qed
\end{proof}

Now we finish the proof of Theorem~\ref{thm:rewards}.
Proposition~\ref{prop:max-comp} and Fact~\ref{fact:reach}
together establish the MDP case.
Establishing the NP $\cap$ coNP upper bound for SSGs 
proceeds in a standard way:
guess a strategy for one player, fix it to get a MDP,
and verify in polynomial time (Proposition~\ref{prop:max-comp}) that
the other player cannot do better than the given value $p$.
To decide whether, e.g., $\Val^{O}(s) \geq p$, guess a strategy 
$\sigma$ for Max, fix it to get an MDP,
and verify that Min has no strategy
$\pi$ so that $\PrA{\sigma,\pi}{s}{O} < p$.
Other cases are similar.
\qed

\noindent Finally, we show that the
upper bound from Theorem~\ref{thm:rewards}
is hard to improve upon:

\begin{proposition}\label{prop:nb-hard}
Assume that a SSG, $\G$, a state $s$, and a reward function $r$
are given, and let $O$ be an objective from $\{\CN,\DI,\PM\}$.
Moreover, assume the property (promise) that
either $\Val^{O}(s) =1$ or $\Val^{O}(s) = 0$.
Then deciding  which is the case
is at least as hard as Condon's~\cite{C92} {\em quantitative} reachability
problem w.r.t.\ polynomial time reductions.
\end{proposition}

\begin{proof}
The problem studied by Condon~\cite{C92} is: given
a SSG, $\H$, an initial state $s$, and a target state $t$, decide
whether
\(
\Val^{\Reach{t}}(s)\geq 1/2
\).
Deciding whether $\Val^{\Reach(t)}(s) > 1/2$ is P-time equivalent.
Moreover,
we may safely assume 
there is a state $t'\not = t$, such that 
whatever strategies are employed, we reach $t$ or $t'$, 
with probability 1.
Consider the following reduction:
given a SSG, $\H$,  with distinguished states $s$, $t$, and $t'$ as
above,  produce a new SSG, $\G$, with rewards as follows:
remove
all outgoing transitions from $t$ and $t'$, 
add transitions
$t \btran{}s$ and $t'\btran{}s$, and
make both $t$ and $t'$ belong to Max.
Let $r$ be the reward function over states of $\G$,
defined as follows: $r(t):=-1$, $r(t'):=+1$ and $r(z):=0$ for all other 
$z\not\in \{t,t'\}$.
It follows from basic random walk theory that in $\G$,
$\Val^{\CN}(s)= 1$ if $\Val^{\Reach{t}}(s)\geq 1/2$, 
and $\Val^{\CN}(s)= 0$ otherwise.
Likewise, 
$\Val^{\DI}(s) =1 $ if $\Val^{\Reach{t'}}(s) > 1/2$,
and $\Val^{\DI}(s)=0$ otherwise,
and identically for the objective $\PM$
which we already showed to be equivalent to $\DI$.
\qed
\end{proof}

\section{Termination}
\label{sec:boundary}

In this section we prove Theorem~\ref{thm:NT-val}.
We continue viewing OC-SSGs as finite-state SSGs
with rewards, as discussed in the introduction.
However, for notational convenience this time we consider rewards 
on \emph{transitions}
rather than on states. It is easy to observe that Theorem~\ref{thm:rewards}
remains valid even if we sum rewards on transitions instead of
rewards on states in the definition of $\CN$.
We fix a SSG, $\G$, with state set $V$, and a reward
function $r$.

\begin{lemma}\label{lem:NT-MD}
Assume that $j\geq |V|$. Then
for all states $s$:
$\Val^{\Term{j}}(s)=1$
iff
$\Val^{\CN}(s)=1$.
\end{lemma}

\begin{proof}
If $\G$ is a maximizing MDP, the proposition is true
by results of~\cite[Section 4]{BBEKW10}.
Consider now the general case, when $\G$ is a SSG.
If $\Val^{\CN}(s) = 1$ then clearly $\Val^{\Term{j}}(s) = 1$.
Now assume that $\Val^{\Term{j}}(s) = 1$ and consider the memoryless
strategy of player Min, optimal for $\CN$,
which exists by Proposition~\ref{prop:long-run-MD-det}.
Fixing it, we get a maximizing MDP, in which
the value of $\Term{j}$ in $s$ is, of course, still $1$.
We already know from the above discussion that the
value of $\CN$ in $s$ is thus also $1$ in this MDP.
Since the fixed strategy for Min was optimal, we get that
$\Val^{\CN}(s) = 1$ in $\G$.
Thus, if $\Val^{\Term{j}}(s) = 1$ then $\Val^{\CN}(s)=1$.
\qed
\end{proof}

\vspace{-1em}

\paragraph{Proof of Theorem~\ref{thm:NT-val}.}
For cases where $j\geq |V|$, the theorem follows directly from Lemma~\ref{lem:NT-MD}
and Theorem~\ref{thm:rewards}.
If $j<|V|$ then we have to perform a simple reachability
analysis, similar to the one presented in~\cite{BBEKW10}.
The following SSG, $\G'$, keeps track of the accumulated rewards
as long as they are between $-j$ and $|V|-j$: its set of states is
$V'\coloneqq\{(u,i) \mid u\in V, -j\leq i \leq |V|-j\}$.

States $(u,i)$ with $i\in\{-j,|V|-j\}$ are absorbing,
and for $i\notin\{-j,|V|-j\}$ we have $(u,i)\tran{}(t,k)$
iff $u\tran{}t$ and $k=i+r(u\tran{}t)$.
Every $(u,i)$ belongs to the player who owned $u$.
The probability of every transition $(u,i)\tran{}(t,k)$, $u\in V_P$,
is the same as that of $u\tran{}t$.
There is no reward function for $\G'$, we consider a reachability objective
instead, given by the target set
$R \coloneqq \{(u,-j)\mid u\in V\}\cup\{(u,i)\mid -j\leq i\leq |V|-j, \Val^{\CN}(u)=1\}$.
Finally, let us observe that, by Lemma~\ref{lem:NT-MD},
$\Val^{\Reach{R}}((s,0))=1$
iff
$\Val^{\Term{j}}(s)=1$.
Since the size of $\G'$ is polynomial in the size of $\G$,
Theorem~\ref{thm:NT-val} is proved.
\qed

\begin{proposition}\label{prop:NT-strat}
For all $j>0$, $s\in V$,
there are pure strategies, $\sigma$ for Max, and $\pi$ for Min, such that
\begin{enumerate}
\item\label{NT-Max}
If $\Val^{\Term{j}}(s)=1$ then $\sigma$ is optimal in $s$ for $\Term{j}$.
\item\label{NT-Min}
If $\Val^{\Term{j}}(s)<1$ then $\sup_\tau\PrA{\tau,\pi}{s}{\Term{j}}<1$.
\end{enumerate}
Moreover, $\sigma$ is memoryless, and $\pi$ only uses memory
of size $|V|$.
Such strategies can be computed 
in P-time for MDPs.
\end{proposition}
\noindent
The proof goes along the lines of the proof of Theorem~\ref{thm:NT-val}.
It can be found
in the appendix, Section~\ref{sec:NT-strat}, 
together with an example that shows the memory use in $\pi$ 
is necessary.

\smallskip

Similarly, both $\Val^{\Term{j}}(s)=0$ and $\Val^{\Term{j}}(s)>0$
are witnessed by pure and memoryless strategies for the respective players.
Deciding which is the case is in P-time,
by assigning the random states to player Max,
obtaining a non-stochastic 2-player one-counter game, and 
using, e.g.,~\cite[Theorem~12]{VASS}.
Finally, we note that from Proposition~\ref{prop:nb-hard} and Lemma~\ref{lem:NT-MD}, it follows that:
\begin{corollary}\label{cor:hard}
Given an SSG, $\G$, and reward function $r$,
deciding whether the value of the termination objective $\Term{j}$  equals 1
is at least as hard as Condon's~\cite{C92} {\em quantitative} reachability
problem, w.r.t.\  P-time many-one reductions.
\end{corollary}

\bibliography{bibliography}

\newpage
\appendix

\section{Appendix}

In the entire appendix, when referring to MDPs and SSGs, we
mean {\em finite-state} MDPs and SSGs.

\subsection{Proof of Proposition~\ref{prop:long-run-MD-det}, objectives
$\negDI$, $\PM$ and $\negPM$}
\label{sec:sumbix}

An objective $O$ is \emph{submixing} if for every run
$w=u_1 v_1 u_2 v_2 \cdots u_k v_k \cdots$, such that 
$u_i$ and $v_i$ are finite paths for every $i$,
and such that both
$u=u_1 u_2 \cdots u_k \cdots$ and
$v=v_1 v_2 \cdots v_k \cdots$ are also runs, we have
$w \in O \implies (u\in O \lor v\in O)$.
This notion is taken directly from~\cite{Gimbert-STACS07}, where it has
been defined in a more general setting. (See also \cite[Section 3]{BBEKW10}
for more details.)
By \cite[Theorem~1]{Gimbert-STACS07}, for every maximizing MDP and
every tail submixing objective, $O$,
player Max has a pure and memoryless optimal strategy.

\begin{lemma}
The objective $\negDI$ is a submixing and
tail objective.
\end{lemma}

\begin{proof}
Obviously it is tail. As for the submixing property,
let $\{a_i\}_{i=1}^\infty$ be a sequence of numbers, and
consider an arbitrary splitting of this sequence into two infinite
subsequences $\{b_i\}_{i=1}^\infty$, $\{c_i\}_{i=1}^\infty$.
For $x\in\{a,b,c\}$ we define
\[
L_x\coloneqq\liminf_{n\to\infty}\sum_{i=1}^n x_i \ .
\]
It is easy to verify that
if at least one of $L_b$, $L_c$ is finite, or if they are infinite
with the same sign, then $L_a \geq L_b + L_c$.
In particular, if $L_a<\infty$ then $\min\{L_b,L_c\}<\infty$.
Applying this to the sequences of rewards finishes the proof.
\qed\end{proof}

\noindent
By changing the $\liminf$ to $\limsup$ in the definition
of $\PM$ we obtain a new objective:
\[
\PM_{+}\coloneqq\{w\in \run_{\G}\mid
\limsup_{n\to\infty} \sum_{i=0}^{n-1} r(w(i))/n>0\} \ .
\]

\begin{lemma}\label{lem:PM-sub}
Both $\PM_{+}$ and $\negPM$ are tail and submixing.
\end{lemma}

\begin{proof}
Both are clearly tail.
For the submixing property, let us start with $\PM_{+}$.
Let $A=\{a_i\}_{i=1}^\infty$ be a sequence of numbers, and
consider an arbitrary splitting of this sequence into two infinite
subsequences $B=\{b_i\}_{i=1}^\infty$, $C=\{c_i\}_{i=1}^\infty$.
For a fixed $n\geq 1$ denote by $n_b\leq n$ the number of
elements of $B$ among the first $n$ elements of $A$.
Then, assuming $n_b<n$
\[
\frac{\sum_{i=1}^{n}a_i}{n}
=
\frac{\sum_{i=1}^{n_b}b_i + \sum_{i=1}^{n-n_b}c_i}{n}
=
\frac{\sum_{i=1}^{n_b}b_i}{n_b} \cdot \frac{n_b}{n}
+
\frac{\sum_{i=1}^{n-n_b}c_i}{n-n_b} \cdot \left(1-\frac{n_b}{n}\right) \ .
\]
Consequently,
\[
\frac{\sum_{i=1}^{n}a_i}{n}
\leq
\frac{\sum_{i=1}^{n_b}b_i}{n_b}
\qquad
\text{or}
\qquad
\frac{\sum_{i=1}^{n}a_i}{n}
\leq
\frac{\sum_{i=1}^{n-n_b}c_i}{n-n_b} \ ,
\]
and thus there is $x\in\{b,c\}$ such that
\[
\limsup_{n\to\infty}
\frac{\sum_{i=1}^{n}a_i}{n}
\leq
\limsup_{n\to\infty}
\frac{\sum_{i=1}^{n}x_i}{n} \ .
\]
The proof for $\negPM$ proceeds similarly, only with reversed signs.
\qed \end{proof}

\noindent
Now we show that $\PM$ is equivalent to $\PM_+$ for memoryless strategies.
\begin{lemma}\label{lem:MD-lim-def}
Under a memoryless strategy, $\sigma$, for a  MDP, $\G$, with a reward function, $r$,
for almost all runs, $w$:
\[
\liminf_{n\to\infty} \sum_{i=0}^{n-1} r(w(i))/n
=
\limsup_{n\to\infty} \sum_{i=0}^{n-1} r(w(i))/n \ .
\]
\end{lemma}

\begin{proof}
Fix $\sigma$ to get a Markov chain on the states of $\G$.
Almost all runs visit some bottom strongly connected component (BSCC),
and the above equality
establishes a prefix independent property. We thus safely assume that
$w$ starts in a BSCC, $C$.
On $C$, $\sigma$ induces an irreducible Markov chain, and
applying the Ergodic theorem (see Theorem~1.10.2 from~\cite{Norris98}) finishes the proof.
\qed\end{proof}

\begin{lemma}\label{lem:mean-payoffs}
For every maximizing MDP, there is always a pure and memoryless strategy, $\sigma$,
optimal for $\PM$.
\end{lemma}

\begin{proof}
Choose $\sigma$ to be optimal for $\PM_+$. This is possible, because
$\PM_+$ is a submixing and tail objective.
Observe that since $\PM\subseteq\PM_{+}$, we have
$\Val^{\PM}(s)\leq\Val^{\PM_{+}}(s)$ for all states $s$.
Finally, due to Lemma~\ref{lem:MD-lim-def}, for all states $s$:
\[
\Val^{\PM_{+}}(s)
=
\PrA\sigma{s}{\PM_{+}}
=
\PrA\sigma{s}{\PM}
\leq
\Val^{\PM}(s) \ .
\]
\qed\end{proof}

One may be tempted to believe that all of the objectives we study
are submixing. This is, however, not true for $\DI$ and $\negCN$,
where we have to employ other methods for proving the existence
of pure and memoryless optimal strategies.

\begin{lemma}
The objectives $\DI$ and $\negCN$ are not submixing.
\end{lemma}

\begin{proof}
Consider the following finite sequences $A_k$ over $\{\pm1\}$, parametrized
by $k\geq 1$, and defined inductively by
$A_1\coloneqq +1, -1$, and $A_{k+1}\coloneqq +1, A_k, -1$.
We build an infinite sequence $A=\{a_i\}_{i=1}^\infty$ by concatenating
them, $A\coloneqq A_1, A_2, A_3, \ldots$.
Obviously $\liminf \sum_{i=1}^na_i=0$. %
Now we define two particular subsequences of $A$, denoted by
$B\coloneqq\{b_i\}_{i=1}^\infty$, $C\coloneqq\{c_i\}_{i=1}^\infty$,
so that
\begin{equation}\label{eq:nosub}
\liminf_{n\to\infty}\sum_{i=1}^n b_i=\liminf_{n\to\infty}\sum_{i=1}^n c_i=-\infty \ .
\end{equation}
We do it inductively by saying for every $k\geq 1$, whether the $k$-th
element, $a_k$, of $A$ belongs to $B$, or $C$.
Assume we have already decided for each of the first $k$
elements of $A$ whether it belongs to $B$ or $C$, so that
we have already defined the finite prefixes $b_1,\ldots,b_M$ of $B$, and $c_1,\ldots,c_N$
of $C$.
Set $s_B^i\coloneqq \sum_{j=1}^ib_j$, and similarly
$s_C^i\coloneqq \sum_{j=1}^ic_j$.
If
either $a_{k+1}=-1$ and $\min_{i=1}^{M}s_B^i \geq \min_{i=1}^{N}s_C^i$,
or $a_{k+1}=1$ and $\min_{i=1}^{M}s_B^i < \min_{i=1}^{N}s_C^i$,
then $a_{k+1}$ belongs to $B$, otherwise it belongs to $C$.
It is easy to verify that for every number $m$ we have
some $n$ such that $s_B^n<m$, and
some $n'$ such that $s_C^{n'}<m$.
(In fact, this is the idea behind the construction -- the sequences
$B$ and $C$ take turns in achieving lower and lower
partial sums.)
Thus (\ref{eq:nosub}) is true.
As the sequence $A$ can be easily obtained as a sequence of rewards
associated to a run of a very simple MDP with rewards,
this proves that $\negCN$ is not submixing.

Similarly goes the proof that $\DI$ is not submixing.
Along the lines of the previous proof, just consider the following
modifications:
Take the sequence $A=\{a_i\}_{i=1}^{\infty}$ to be defined by
$a_i=-1$ iff $i \equiv 0 \pmod 3$ and
$a_i=+1$ otherwise.
Further, in the inductive process of building the sequences $B$ and $C$,
denote by $z_B\coloneqq |\{i\leq M \mid s_B^i=0\}|$, 
and by $z_C\coloneqq |\{i\leq N \mid s_C^i=0\}|$.
Finally, apply the rule of assigning $a_{k+1}$ to $B$ iff
either $a_{k+1}=-1$, $z_C\geq z_B$, and $s_B^M>0$, or
$a_{k+1}=+1$ and $z_C<z_B$ or $s_B^M=0$.
(Here the intuition is that $B$ and $C$ take turns in revisiting
$0$ from above.)
It is easy to show that for every $m\geq 0$ there is some
$n\geq m$ such that $s_B^n=0$, and some $n'\geq m$ such that $s_C^{n'}=0$.
This shows that $\liminf B=\liminf C=0$, while $\liminf A=\infty$.
As a consequence, $\DI$ is not submixing.
\qed\end{proof}

\subsection{Proof of Proposition~\ref{prop:long-run-MD-det}, objective $\DI$}

First we set up a tool to analyze finite-state Markov chains with respect to
the objective $\DI$.
Consider a finite-state Markov chain, $\M$, with the underlying transition
graph $(S,\ctran{})$, and with a reward function, $r:S\to\{-1,0,+1\}$.
Assume, moreover, that $\M$ is irreducible.
Also assume that some initial state, $s$, is fixed.
We derive here one condition sufficient for $\PrA{}{s}{\DI}=1$,
and another one sufficient for $\PrA{}{s}{\DI}=0$ in $\M$.
The conditions are parametrized by a choice of a subset $R\subseteq S$ of the states
of $\M$.
To formulate them we need the following random variables.
\begin{itemize}
\item $V_k^t$, $k\geq0$, $t\in R$ returns the time of the $k$-th visit (thus ``V'') to $t$.
\item $G_k^t$, $k\geq0$, $t\in R$ is the reward gained (``G'') between
time $V_k^t$ (inclusive) and the next visit to $R$ (exclusive).
\end{itemize}
By standard facts from probability theory, almost all runs in $\M$
visit all states infinitely often. Thus these random variables are almost
surely defined.
For a fixed $t\in R$, all the variables $G_k^t$ are i.i.d., and,
as the expected time to visit $R$ from $t$ is finite, their common
mean, $\mu_t$, is well defined and finite.
Observe also that the values $\mu_t$ do not depend on the choice
of the initial state.

\begin{lemma}\label{lem:MC-classify}
For every finite-state irreducible Markov chain, $\M$, and
every subset, $R$, of states, and every $t\in R$,
considering the numbers $\mu_s$, $s\in R$,
derived as above, the following is true:
\begin{itemize}
\item If $\mu_s>0$ for all $s\in R$ then $\PrA{}{t}{\DI}=1$.
\item If $\mu_s\leq0$ for all $s\in R$ then $\PrA{}{t}{\DI}=0$.
\end{itemize}
\end{lemma}

\begin{proof}
We use the following random variables on runs from $\run_\M(t)$:
\begin{itemize}

\item $V_k$, $k\geq1$, the time of the $k$-th visit to $t$.
(Note: $V_1\equiv 0$.)

\item 
$A_k$, $k\geq1$,
the reward accumulated (``A'') between time $V_k$ (inclusive)
and $V_{k+1}$ (exclusive).

\item $S_m\coloneqq\sum_{k=1}^m A_k$, $m\geq0$. (``S'' for ``sum''. Note:
$S_0\equiv 0$.)

\end{itemize}
Since $\M$ is a Markov chain,
we get that the variables $A_k$ are i.i.d., in particular
there is some $\mu$ such that $\mu=\Ex{}{t}{A_k}$
for all $k\geq 1$.
\footnote{By $\ExSign$ we denote the expectation.}

\begin{claim}
If all $\mu_s>0$ then $\mu>0$.
If all $\mu_s\leq0$ then $\mu\leq0$.
\end{claim}

\begin{proof}
For every $s\in R$ and $w\in\run_\M(t)$ denote by $v_s(w)$
the number of visits to $s$ before the first revisit to $t$:
\(
v_s(w)=card\left(
\{k\geq 0 \mid w(k)=s \land \forall l<k: w(l)=t \implies l=0\}
\right) %
\).
Then, writing $R=\{t_1,\ldots,t_\ell\}$,
\[
\mu =
\sum_{c_1,\ldots,c_\ell\geq0}
\PrA{}{t}{\bigwedge_{j=1}^\ell v_{t_j}=c_j}
\cdot
\sum_{j=1}^\ell c_j\cdot\mu_{t_j} \ .
\]
Since all the coefficients of $\mu_s$, $s\in R$
are non-negative, the claim is proved.
\qed
\end{proof}

For $\mu\leq 0$ standard results on random walks
(see~\cite[Theorem~8.3.4]{Chung01}) yield
$\liminf_{n\to \infty} S_n<\infty$
almost surely.
Immediately,
$\liminf_{n\to \infty}\sum_{i=0}^nr(w(i))<\infty$
almost surely, thus $\PrA{}{t}{\DI}=0$.

\noindent
The case when $\mu>0$ is more subtle, and we need to introduce two more random variables:
\begin{itemize}
\item
$M$, the least $m$ such that $S_m>0$. (``M'' for ``maximum''.)

\item
$M'$, $M'\coloneqq V_M$ (the actual number of steps to $M$).

\end{itemize}

\begin{claim} {\bf (cf.\ \cite[Theorem~8.4.4]{Chung01})}
$\Ex{}{t}{M}<\infty$.
\end{claim}

\begin{claim}
\(
\Ex{}{t}{V_{k+1}-V_k}=\Ex{}{t}{V_2-V_1}<\infty
\)
for all $k\geq1$.
\end{claim}

\begin{proof}
Since $\M$ is a Markov chain, we get the equality.
By standard results on Markov chains
(see~\cite[Theorem~1.7.7]{Norris98})
we obtain 
$\Ex{}{t}{V_2-V_1}=\Ex{}{t}{V_2}=(\pi(t))^{-1}$ where $\pi$ is 
an invariant (and positive)
distribution over the states of $\M$.
Thus the inequality follows.
\qed
\end{proof}

\begin{claim}\label{claim:Mprime-fin}
$\Ex{}{t}{M'}<\infty$.
\end{claim}
\begin{proof}

\begin{align*}
\Ex{}{t}{M'}
&=\sum_{m=1}^\infty \PrA{}{t}{M=m}\cdot \Ex{}{t}{V_m} \\
&=\sum_{m=1}^\infty \PrA{}{t}{M=m}\cdot \Ex{}{t}{(V_m-V_{m-1})+(V_{m-1}-V_{m-2})+\cdots+(V_2-V_1)} \\
&=\sum_{m=1}^\infty \PrA{}{t}{M=m}\cdot (m-1)\cdot \Ex{}{t}{V_2-V_1} \\
&=(\Ex{}{t}{M}-1)\cdot \Ex{}{t}{V_2-V_1}
\end{align*}
\qed
\end{proof}

As a generalization of the variable $M$,
we define, inductively and for almost all runs from $\run_\M(t)$,
yet another sequence $M_k$, $k\geq0$ of random variables by setting
$M_0\equiv 0$, and
$M_{k+1}$ to be the least $m$ such that $S_m>S_{M_k}$.
(We get $M=M_1$.)
In other words, $M_k$ are the times when maximal rewards were achieved
on revisit to $t$.
We also define a sequence of events, $Z_k$, $k\geq 1$:
A run $w\in \run_\M(t)$ is in $Z_k$ iff there is some
$j$, $V_{M_k}\leq j < V_{M_{k+1}}$ such that
the reward accumulated on $w(0)\cdots w(j)$ is $0$.
(``Z'' for ``zero''.)

\begin{claim}
$\sum_{k\geq1}^\infty \PrA{}{t}{Z_k}<\infty$.
\end{claim}
\begin{proof}
It takes at least $S_{M_k}\geq k$ steps to gain
reward $0$ starting at time $V_{M_k}$.
Since $V_{M_{k+1}}-V_{M_k}$ has the same distribution as $M'$,
we get
\(
\PrA{}{t}{Z_k}
\leq\PrA{}{t}{M'\geq k}
\).
Now
\[
\sum_{k\geq1}^\infty \PrA{}{t}{Z_k}
\leq \sum_{k\geq1}^\infty \PrA{}{t}{M'\geq k}
= \sum_{k\geq1}^\infty \sum_{l\geq k}^\infty \PrA{}{t}{M'= k}
= \sum_{k\geq1}^\infty k \cdot \PrA{}{t}{M'= k}
= \Ex{}{t}{M'} < \infty \ .
\]
\qed
\end{proof}

Thus
by the Borel-Cantelli lemma, the probability that $Z_k$ occurs for
infinitely many $k$ is $0$.
Consequently $\liminf_{n\to\infty} \sum_{i=0}^nr(w(i)) > 0$ for almost all $w$.
Similarly we can prove for all $h>0$ that
$\liminf_{n\to\infty} \sum_{i=0}^nr(w(i)) > h$ for almost all $w$.
Hence,
$\liminf_{n\to\infty} \sum_{i=0}^nr(w(i)) = \infty$ almost surely,
because a countable intersection of sets of probability $1$ has probability $1$.
Thus
$\PrA{}{t}{\DI}=1$.
\qed\end{proof}

\begin{lemma}\label{lem:fin-opt}
The finite-memory strategy $\sigma$ from the proof of
Proposition~\ref{prop:long-run-MD-det} is optimal for $\DI$.
\end{lemma}

\begin{proof}
Observe that fixing $\sigma$ yields a finite-state
Markov chain, $\G(\sigma)$, on the parallel composition of $\G$ and
the finite automaton used for updating the memory of $\sigma$.
Let us fix an arbitrary bottom strongly connected component (BSCC), $C$,
of $\G(\sigma)$, and denote by $R$ the states of $C$ in which the memory
of $\sigma$ is being reset.
We are now going to analyze, using Lemma~\ref{lem:MC-classify},
the irreducible MC, $\M$, induced
by restricting $\G(\sigma)$ to $C$.
Fix an arbitrary $s\in R$.
Recall, that the variable $G_k^u$, defined before stating Lemma~\ref{lem:MC-classify},
returns the reward accumulated between the $k$-th visit to $s$ and
the next visit to $R$.
It is easy to verify that the common mean, $\mu_u$, of $G_k^u$
is equal to the mean of the stopping time $T_s$ introduced
in the main text of the proof, ant thus positive.
Therefore Lemma~\ref{lem:MC-classify} guarantees that
for every state $s\in R$ lying in some BSCC we have
$\PrA{}{s}{\DI}=1$.
Since $\G(\sigma)$ is finite, almost every run in it reaches
some BSCC and every state in it.
Because $\DI$ is a tail objective we get
$\PrA{}{s}{\DI}=1$
for every state $s$.
\qed\end{proof}

\begin{lemma}\label{lem:DI-MD}
In a maximizing MDP, $\G$, with value $1$ in all states,
given a pure finite-memory strategy $\sigma$ optimal for $\DI$,
a pure and memoryless optimal strategy $\tau$ can be constructed.
\end{lemma}

\begin{proof}
As in the proof of Lemma~\ref{lem:fin-opt},
given a finite-memory strategy, $\varrho$, we denote
by $\G(\varrho)$ the finite-state Markov chain, states
of which are
pairs $(s,q)$ where $s$ is a state of $\G$, and
$q$ is a state of the finite automaton representing the memory of $\varrho$.
Probabilities are obtained in the natural way from $\varrho$ and $\G$.
Consider now the Markov chain $G(\sigma)$.
The initial state is $(s_0, q_0)$ where $s_0$, $q_0$ are initial
states of $\G$, and the automaton for $\sigma$, respectively.
For technical reasons we assume that for each $q$ there is at most
one $s$ so that $(s,q)$ is reachable from $(s_0,q_0)$.

If there are two states, $q\neq p$, of the automaton for $\sigma$,
and a state $s$ of $\G$
such that both $(s,q)$ and $(s,p)$ are reachable from $(s_0,q_0)$,
we call both $q$ and $p$ \emph{ambiguous}.
If there is no ambiguous state, $\sigma$ is already memoryless.
If there are ambiguous states, we show how to
modify $\sigma$ to get another pure and finite-memory optimal strategy $\sigma'$,
such that the associated Markov chain, $\G(\sigma')$, has
fewer ambiguous states.
As there are only finitely many ambiguous states in the beginning,
repeating this process inevitably leads to the optimal pure and
memoryless strategy $\tau$.

We thus assume that there is a state $s$ of $\G$ such that
$A\coloneqq\{(s,q) \mid (s,q)\text{ is reachable from }(s_0, q_0)\}$
has at least two elements.
For every fixed choice of $(s,q)\in A$ we now define a new finite-memory strategy
$\sigma_q$.
This is derived by modifying the finite automaton
for $\sigma$ so that all transitions leading to some $p$, where
$(s,p)\in A$, are redirected to $q$.
From this, due to our technical assumption,
already follows that $\sigma'$ has fewer ambiguous states.
It remains to prove that there is some $q$
such that $\sigma_q$ is optimal.

There are two cases to consider.
First, consider the situation where there is $(s,q)\in A$
such that with some positive probability states from $A\smallsetminus \{(s,q)\}$
are visited only finitely often in $\G(\sigma)$.
This implies that there is a BSCC, $S$, of $\G(\sigma)$, such that $|S\cap A|\leq1$.
We choose $(s,q)$ so that it minimizes the distance
(in the transition graph of $\G(\sigma)$)
to $S$ among the states from $A$.
This implies that, starting in $(s,q)$, states
from $A\smallsetminus \{(s,q)\}$ are avoided with some positive probability, $\delta$.
We now prove that $\sigma_q$ is optimal.
Indeed, let $\neg A$ be the event of not visiting $A$, and
let $E$ be an arbitrary event.
Then $\PrA{\sigma}{(s_0,q_0)}{E \mid \neg A}=\PrA{\sigma_q}{(s_0,q_0)}{E \mid \neg A}$.
On the other hand, every run in $\G(\sigma)$ visiting $A$ projects to $\G(\sigma_q)$,
as a run $w$ visiting $(s,q)$.
Here we have two possibilities. Either $\delta=1$, and 
we set $w_q$ to be the suffix of $w$ starting with the first occurrence of $(s,q)$.
Or $\delta<1$, implying that $S\cap A=\emptyset$ and thus
$(s,q)$ is not in a BSCC. Thus on almost all runs $(s,q)$ is visited
only finitely many times, and we may define $w_q$ to be the suffix
starting with the last occurrence of $(s,q)$ in $w$.
For every event $E$
we define the set $E'\coloneqq\{w\in\run_\G \mid w_q \in E\}$.
Denoting simply by $A$ the event of visiting $A$,
it is easy to verify for all $E$ that
$\PrA{\sigma}{(s_0,q_0)}{E \mid A} = \PrA{\sigma_q}{(s_0,q_0)}{E' \mid A}$.
Since $\DI$ is tail, we have $\DI'\subseteq\DI$. Thus almost
all runs in $\G(\sigma_q)$ satisfy $\DI$.

If the first case does not apply then there must be a BSCC, $S$, such that
$|S\cap A|\geq2$. Using Lemma~\ref{lem:MC-classify} for $\G(\sigma)$ restricted
to $S$, with $R=A$, we obtain that there must be at least
one $(s,q)\in A$ such that the expected accumulated reward until revisiting $A$,
$\mu_q$, is positive.
Observe that $(S\smallsetminus A)\cup\{(s,q)\}$ forms a BSCC, $S_q$, in $\G(\sigma_q)$.
Using Lemma~\ref{lem:MC-classify} on $\G(\sigma_q)$ restricted to $S_q$, with
$R=\{(s,q)\}$, we obtain that all runs in $\G(\sigma_q)$ started
in $S_q$ satisfy $\DI$. Because, similarly to the previous case, almost all runs in
$\G(\sigma_q)$ remain either unaffected or visit $S_q$, we obtain again that almost
all runs in $\G(\sigma_q)$ satisfy $\DI$.
\qed\end{proof}

\subsection{Proof of Proposition~\ref{prop:long-run-MD-det}, objective $\negCN$}

Recall that
$W_\infty = \{ s \mid \Val^{\DI}(s)=1\}$, and
$W_+ = \{ s \mid \Val^\Nn(s)=1\}$.

\begin{lemma}\label{lem:CN-di-nn}
For every maximizing MDP, $\G$, and its state, $s$,
if $\Val^\negCN(s)=1$ then
there is a strategy, $\tau$, such that
$\PrA\tau{s}{\Reach{W_\infty\cup W_+}}=1$.
\end{lemma}

\begin{proof}
If $\Val^{\Reach{W_\infty}}(s)=1$ then we are already done for this state.
Assume $\Val^{\Reach{W_\infty}}(s)<1$.

\begin{claim}
There is a strategy, $\tau$, such that
\begin{enumerate}
\item\label{enum:winDI} $\PrA\tau{s}{\negCN}=1$;
\item\label{enum:memless} $\tau$ restricted to $W_\infty$ is memoryless and
$\PrA\tau{v}{\DI}=1$ for all $v\in W_\infty$;
\item\label{enum:between} almost all $w\in\run_{\G(\tau)}(s)$ which
do not visit $W_\infty$ satisfy

\begin{equation}\label{eq:liminf-bound}
\infty > \liminf_{n\to \infty} \sum_{i=0}^{n}r(w(i)) > -\infty \ .
\end{equation}
\end{enumerate}
\end{claim}

\begin{proof}
Choose a strategy satisfying~\ref{enum:winDI}, it must exist by~\cite[Theorem~3.1]{GH-SODA10}.
By Proposition~\ref{prop:long-run-MD-det} for $\DI$ we obtain~\ref{enum:memless},
because $\DI\subseteq\negCN$.
On the other hand, runs
avoiding $W_\infty$ belong to $\DI$ with probability $0$, as a consequence of
Lemma~\ref{lem:reach-mdp-qual} for the objective $\DI$,
\footnote{A careful reader may suspect a circular dependency since Lemma~\ref{lem:reach-mdp-qual}
uses Proposition~\ref{prop:long-run-MD-det}. This is, however, a correct use, since
it only uses the proposition for $\DI$, which has already been proved.}
proving~\ref{enum:between}.
\qed
\end{proof}

We now define an event $\Inf{v}$ for all states $v$.
Consider a run $w$, satisfying (\ref{eq:liminf-bound}).
There must be
some integer $\ell$ such that $\sum_{i=0}^{n}r(w(i))\geq\ell$ for all $n\geq0$.
Choosing the greatest such $\ell$, there is some index, $j$, such that
$\sum_{i=0}^{j}r(w(i))=\ell$.
We call the smallest such $j$ to be the \emph{minimum} of $w$,
and $\ell$ is said to be the \emph{minimal value} of $w$.
According to this, we define inductively the following functions:
$M_1(w)$ is the minimum of $w$, and,
given $n\coloneqq M_k(w)$, and the suffix
$w'=w(n{+}1)\, w(n{+}2) \cdots$ of $w$, we set
$M_{k+1}(w)\coloneqq M_1(w')+n+1$ for $k\geq1$.
Further, $m_k\coloneqq \sum_{i=0}^{M_k(w)}r(w(i))$.
(See also Figure~\ref{fig:minima} for an example.)
The sequence $\{m_k\}_{k=1}^\infty$ is non-decreasing and,
due to the first inequality in (\ref{eq:liminf-bound}) also
bounded, hence it has a well defined finite limit, $\bar{m}$.
Given some state $v$, we define an event, $\Inf{v}$, by the condition that
there are infinitely many
$k$ such that the state visited at time $M_k$ is $v$ and $m_k=\bar{m}$.

\begin{figure}
\begin{center}
\begin{tikzpicture}[x=0.3cm,y=0.3cm]
\draw[->] (0,-3.5) -- (0,7);
\draw[->] (0,0) -- (13,0);
\draw[dotted] (0,-2) -- (13,-2);
\draw[dotted] (0,4) -- (13,4);
\draw[dotted] (2,-2) -- (2,0);
\draw[dotted] (11,4) -- (11,0);
\draw (1.0,7.8) node {$\sum_{i=0}^{k}r(w(i))$};
\draw (13.5,0) node {$k$};
\draw (2,1) node {$M_1$};
\draw (11,-1) node {$M_2$};
\draw (-1.7,0) node {$(0,0)$};
\draw (-1.7,4) node {$m_2$};
\draw (-1.7,-2) node {$m_1$};
\foreach \x/\y in {0/-1,1/-1,2/-2,3/-1,4/0,5/1,6/2,7/3,8/4,9/5,10/5,11/4,12/5,13/6}
{
\fill (\x,\y) circle (0.2);
}
\end{tikzpicture}
\end{center}
\caption{An example of a run and its minima.}\label{fig:minima}
\end{figure}
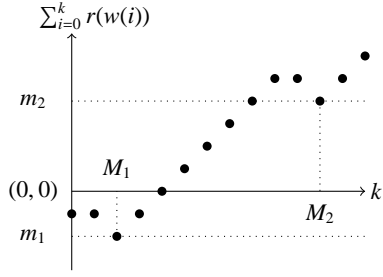

\begin{claim}
For all states $v$,
if $\PrA\tau{v}{\Inf{v}}>0$
then $v\in W_+$.
\end{claim}

\begin{proof}
Fix a state $v$ satisfying the assumption of the claim.
Note that due to our choice of $\tau$,
for all such $v$, $W_\infty$ is not reached on a run from $\Inf{v}$.
Observe that $\Inf{v}$ is tail, so by \cite[Theorem~3.1]{GH-SODA10} there is
a state $v'$ and a strategy $\pi$ such that $\PrA\pi{v'}{\Inf{v}}=1$.
In particular, this must be true for $v'=v$ since $\Inf{v}\subseteq\Reach{v}$,
and the objective
is tail. Finally, the strategy can be chosen so that almost surely
$m_k=0$ for all $k\geq1$.
In other words, 
$\PrA\pi{v}{\Nn}=1$. In particular, $v\in W_+$.
\qed
\end{proof}

Since there are only finitely many states, the union $\bigcup_{v}\Inf{v}$
has probability $1$ on the condition of not reaching $W_\infty$.
The last claim showed that
$\PrA\tau{s}{\Reach{W_+} \mid \Inf{v}}=1$
for all $v\in V$ with $\PrA\tau{s}{\Inf{v}}>0$.
This proves
$\PrA\tau{s}{\Reach{W_\infty\cup W_+}}=1$.
\qed\end{proof}

\subsection{Proof of Lemma~\ref{lem:01bscc}}

\begin{reftheorem}{Lemma}{\ref{lem:01bscc}}
\lemmabscctext{}
\end{reftheorem}
\smallskip

\begin{proof}
From every state, $s$, every other state, $t$, is visited almost
surely. $O$ is tail, thus $\PrA{}{s}{O}=\PrA{}{t}{O}$.
Assume that $\PrA{}{s}{O}>0$ for some $s$, and thus for all $s$. 
Since a Markov chain is a special case of a SSG, we directly
apply~\cite[Theorem~3.2]{GH-SODA10} and get that $\PrA{}{s}{O}=1$.
\qed
\end{proof}

\subsection{Proof of Lemma~\ref{lem:PM-DI-eq}}

\begin{lemma}\label{lem:PM-DI-scc}
Let $\M$ be an irreducible Markov chain with rewards on states,
and $s$ a fixed state of $\M$.
If $\PrA{}{s}{\DI}=1$ then $\PrA{}{s}{\PM}=1$.
\end{lemma}

\begin{proof}
We fix $s$ as a starting state.
Denote by $X_k$, $k\geq1$ the reward accumulated between
the $k$-th (inclusive) and $k{+}1$-st (exclusive) visit to $s$.
Since $\M$ is a Markov chain, these variables are i.i.d.; we denote by
$\mu$ their common mean.
Choosing $R=\{s\}$ in Lemma~\ref{lem:MC-classify} yields that
$\mu>0$.
Thus the sums, $S_\ell\coloneqq\sum_{k=1}^\ell X_k$ define a homogeneous
random walk with a positive drift.
Define:
\begin{itemize}
\item $V_k$, $k\geq1$ to be the time of the $k$-th visit to $s$ (note: $V_1\equiv0$),
\item $M$ to be the least $k$ such that $S_k>0$, and
\item $M'\coloneqq V_M$.
\end{itemize}
By Claim~\ref{claim:Mprime-fin} from the proof of Lemma~\ref{lem:MC-classify}
we know that $\Ex{}{s}{M'}<\infty$.
Further we define:
\begin{itemize}
\item $M_0\equiv 0$, and $M_k$, $k>0$ to be the least
$m$ such that $S_m>S_{M_{k-1}}$, and
\item $Y_k \coloneqq V_{M_{k}+1} - V_{M_{k-1}+1}$, $k>0$.
(Note: $\sum_{k=1}^n Y_n=V_{M_n+1}$.)
\end{itemize}
The variables in the sequence $Y_k$ are independent
and distributed identically with $M'$, thus we may
apply the strong law of large numbers
(see, e.g., Theorem~1.10.1 in~\cite{Norris98})
and obtain that almost surely
\[
\lim_{n\to\infty} \frac{\sum_{k=1}^n Y_n}{n}
=\lim_{n\to\infty} \frac{V_{M_n+1}}{n}
= \Ex{}{s}{M'} < \infty \ .
\]
Because $S_{M_n+1}\geq n$,
we have almost surely
\[
\lim_{n\to\infty} \frac{S_{M_n+1}}{V_{M_n+1}}
\geq
\lim_{n\to\infty} \frac{n}{V_{M_n+1}}
= \frac{1}{\Ex{}{s}{M'}} > 0 \ .
\]
Because the leftmost term is equal to the mean payoff, we conclude that
$\PrA{}{s}{\PM}=1$.
\qed\end{proof}

\subsection{Proof of Proposition~\ref{prop:NT-strat}}
\label{sec:NT-strat}

Recall the SSG $\G'$ with the reachability objective $R$
from the proof of Theorem~\ref{thm:NT-val}.
This game emulates playing $\G$ until (1) the accumulated
reward exceeds $|V|-j$ or a state $u$ with $\Val^\CN(u)=1$ is visited --
then, by Lemma~\ref{lem:NT-MD}
the players may switch to optimizing the probability of $\CN$ instead --
or until (2) the accumulated reward is $-j$.
Memoryless strategies for $\G'$ induce strategies for $\G$ which use
memory of size $|V|$ to store the accumulated reward until it exceeds $|V|-j$
or hits $-j$.
From this and from the analysis in the proof of Theorem~\ref{thm:NT-val}
we can see that the strategies $\sigma$ and $\pi$ from the statement of the proposition,
are easy to construct, with the promised time complexity,
to be pure and using only a finite memory of size $|V|$.

The last thing to show is how to transform $\sigma$
to some memoryless $\sigma'$, preserving the optimality for $\Term{j}$ in $s$.
Restricted to states $u$ with $\Val^\CN(u)=1$, $\sigma$ is already
memoryless.
Call these $u$ \emph{safe}. We set $\sigma'(u)=\sigma(u)$ for every safe $u$.
We further call \emph{unsafe} those states $u$
which are not safe, but there is some strategy $\tau$
such that $\PrA{\sigma,\tau}{s}{\Reach{u}}>0$.
Unsafe states may have been visited with various accumulated
rewards so far, but from what we already proved it follows that
all these accumulated rewards lie between $-j$ and $|V|-j$ (excl.).
For an unsafe $u$, denote by $i_u$ the maximal such accumulated reward,
and by $w_u$ some history along which this was accumulated.
It remains to define $\sigma'$ for unsafe $u$.
We simply set $\sigma'(u)=\sigma(w_u)$.
Since, under $\sigma'$, no unsafe state is reached from
a safe state, $\sigma'$ is still optimal for $\Term{j}$
in all safe states.
Consider an unsafe $u$, and some $i$, $-j<i\leq i_u$, and an arbitrary
strategy $\pi'$ for Min.
Then in $\G$, under the strategies $(\sigma',\pi')$, on condition that $u$ was visited with
an accumulated reward $i$, almost all runs from $u$ either visit a
safe state,
or the accumulated reward reaches $-j$ at some point,
or an unsafe state $t$ is visited, and at the same time
the accumulated reward is at most $i_t+i-i_u$.
Thus by double induction, first on
$|V|-j-i_u$ then on $i$, for all unsafe $u$ and $i\leq i_u$ we have
that $\sigma$ is optimal for $\Term{i}$ in $u$.
Thus $\sigma$ is pure, memoryless and optimal for $\Term{j}$ in $s$.
\qed

\paragraph{An example where memory for Min is needed.}
This example shows that the strategy $\pi$ of player Min from Proposition~\ref{prop:NT-strat}
may indeed have to use memory.
Consider this minimizing MDP:
\begin{itemize}
\item States: $v,low,up,back,down$; Min owns $V_\bot=\{v\}$.
\item Transitions (and their rewards):
$v\tran{}back$ (reward $0$), $v\tran{}low$ ($-1$),
$low \tran{}up$ ($+1$), $up \tran{}up$ ($+1$),
$back \tran{} down$ ($0$), $back \tran{} v$ ($+1$),
$down\tran{}down$ ($-1$).
\item From probabilistic states the successor is chosen uniformly among available transitions.
\end{itemize}
Then for all $j>0$: $\Val^{\Term{j}}(v)<1$, as witnessed by the strategy
choosing $back$ as a successor of $v$ whenever the reward accumulated so far
is $1$, and $low$ in all other cases.
However, there are only two pure and memoryless strategies for Min:
\begin{itemize}
\item choosing the transition $v\tran{}back$ makes the probability
of $\Term{j}$ to be $1$ for all $i > 0$;
\item choosing $v\tran{}low$ makes the probability of $\Term{1}$ to be $1$.
\end{itemize}

\end{document}